%% file: main.tex
\documentclass[11pt]{article}
\usepackage{graphicx,fullpage} 
\usepackage{amsfonts,amsmath,amsthm,amssymb,mathtools, comment}\usepackage{xcolor}
\usepackage[T1]{fontenc}
\usepackage[colorlinks,citecolor=blue,linkcolor=blue,urlcolor=red,pagebackref]{hyperref}

\usepackage[capitalise]{cleveref}
\usepackage{url}
\usepackage{thmtools, thm-restate}
\usepackage{algorithm, algorithmicx, algpseudocode}
\usepackage{tikz}
\usepackage{bm}
\usepackage{thm-restate}

\newcommand{\eps}{\varepsilon}

\newcommand{\patrascu}{P\u{a}tra\c{s}cu}

\newtheorem{theorem}{Theorem}[section]  
\newtheorem{lemma}[theorem]{Lemma}
\newtheorem{proposition}[theorem]{Proposition}
\newtheorem{claim}[theorem]{Claim}

\theoremstyle{definition}

\title{New Diameter Approximations via Distance Oracle Techniques}

\author{Yael Kirkpatrick\thanks{\texttt{yaelkirk@mit.edu}, supported by NSF Grant No 2141064.}\\MIT \and Liam Roditty\thanks{\texttt{liam.roditty@biu.ac.il}, supported by BSF Grant 2024233.}\\Bar Ilan University \and Richard Qi\thanks{\texttt{rqi@mit.edu}.}\\MIT \and Virginia Vassilevska Williams\thanks{\texttt{virgi@mit.edu}, supported by NSF Grant CCF-2330048, BSF Grant 2024233 and a Simons Investigator Award.}\\MIT}
\date{}

\begin{document}

\maketitle

\begin{abstract}
    Computing the diameter of a graph is a problem of great interest both in general algorithms research and specifically within fine-grained complexity, where it is a cornerstone hard problem. As computing the exact diameter in $m$-edge graphs requires $m^{2-o(1)}$ time under the Strong Exponential Time Hypothesis, much work has gone into approximating this parameter. Recent work has achieved a full conditional lower bound tradeoff curve for both directed and undirected graphs [Dalirrooyfard, Li and Vassilevska W., FOCS'21]. However, the best known upper bounds do not match the lower bounds. In particular, the best known approximation scheme for undirected graph diameter [Cairo-Grossi-Rizzi, SODA 2016] has not been improved. Moreover, this scheme is randomized and no similar deterministic scheme is known.

    Another fundamental field of research in shortest paths computation is the construction of approximate distance oracles. Thorup and Zwick [JACM'05] provided the first such distance oracle with constant query time and (conditionally) optimal space, and in the years since many advances have led to a vast toolbox of techniques and data structures.

    These two areas of research seem natural to combine since they both concern approximating shortest paths. However, the known diameter approximation algorithms only use a small subset of the techniques used in distance oracles research. In this work we show that in fact approximate diameter and distance oracles are intricately connected.

    We first demonstrate a strong connection between the current best known diameter approximation scheme of Cairo, Grossi and Rizzi ("CGR") and the $(2k-1)$-approximate distance oracle of Thorup and Zwick. This allows us to derandomize the CGR algorithm and obtain the first deterministic diameter approximation tradeoff.

    We further derandomize other central techniques in the field of distance oracles and use them to achieve new deterministic diameter approximation algorithms, including a simpler $3/2$-approximation with no additive error and a new $5/3$-approximation, the first new step in the diameter approximation tradeoff in almost a decade. Finally, we show how these new techniques can be used to derandomize many current best known results in various fields of shortest paths approximations.
\end{abstract}

\newenvironment{proofof}[1]
  {\begin{proof}[#1]}
  {\end{proof}}

    \input{1_introduction}
    \input{2_prelim}
    \input{3_deterministic_cgr}
    \input{4_deterministic_small_clusters}

    \input{5_additional_diam_approx}
    \input{6_deterministic_cgr_radius_ecc}

    \bibliographystyle{abbrv} 
    \bibliography{ref}
\end{document}

%% file: 1_introduction.tex
\section{Introduction}
One of the most fascinating problems in shortest paths algorithms is estimating the {\em diameter} of a given graph. The diameter is the largest shortest path distance and is a natural parameter that measures how fast information can spread in a network. Computing the diameter can be accomplished by first computing All-Pairs Shortest Paths (APSP), going through all computed distances and returning the largest.
This approach is largely unsatisfying because the output size of APSP in $n$-node graphs is $n^2$ and hence even in very sparse graphs, one cannot hope to obtain sub-$n^2$ time algorithms for the diameter if one first computes APSP. The situation is even worse in dense weighted graphs, where the fastest exact algorithm for APSP runs in $n^3/\exp(\sqrt{\log n})$ time \cite{ryanapsp}. The fastest approximation algorithms use fast matrix multiplication: Zwick \cite{zwickbridge} provided an $\tilde{O}(n^{\omega}/\eps)$ time \footnote{$\tilde{O}$ subsumes polylogarithmic factors.} $(1+\eps)$-approximation algorithm where $\omega<2.372$ is the exponent of square matrix multiplication \cite{AlmanDVXXZ25}. (Boolean) matrix multiplication is also known \cite{doi:10.1137/S0097539797327908} to be necessary for any finite directed (or factor $<2$ undirected) approximation algorithm for APSP.

It remains a big open problem whether computing the diameter of a graph exactly requires computing APSP. Roditty and Vassilevska W. \cite{rv13} showed that the Strong Exponential Time Hypothesis (SETH) (\cite{ip1,cip10}) implies that computing the exact diameter requires $m^{2-o(1)}$ time in sparse graphs.
Because of this, a lot of research has focused on obtaining fast, subquadratic time approximation algorithms.

Aingworth, Chekuri, Indyk and Motwani \cite{aingw} gave an $\tilde{O}(n^2+m\sqrt n)$ time deterministic algorithm that achieves an ``almost'' $3/2$ approximation\footnote{Here ``almost'' means that there is extra small additive error, in addition to the multiplicative approximation factor.}: an estimate $\tilde{D}$ of the diameter $D$ such that $2D/3-M/3\leq \tilde{D}\leq D$ in any directed or undirected graph with nonnegative integer edge weights bounded by $M$.
Roditty and Vassilevska W. \cite{rv13} obtained the same guarantees as \cite{aingw} but with improved (expected) running time $\tilde{O}(m\sqrt n)$ {\em at the cost of randomization}.

Cairo, Grossi and Rizzi \cite{cgr} extended the techniques of \cite{rv13,aingw} to obtain a runtime/approximation tradeoff for undirected graphs. For every $k\geq 1$, they construct an algorithm with {\em expected} running time
$\tilde{O}(mn^{1/k})$ achieving $2^{k-1} D/(2^k-1)-(2^{k-1}-1)M/(2^k-1)\leq \tilde{D}\leq D$, an almost $(2-1/2^{k-1})$-approximation. We call this approximation scheme ``CGR''. Abboud et al. \cite{directeddiam} obtain a weaker scheme for directed graphs.

Chechik et al. \cite{chechikdiam} removed the additive error in the almost $3/2$-approximation algorithm of \cite{rv13} at the cost of  higher running time $\tilde{O}(\min\{mn^{2/3},m^{3/2})$. Backurs et al. \cite{towardsdiam} provide an $\tilde{O}(n^2)$ almost-$3/2$ approximation and an $O(n^{2.05})$ time almost-$3/2$ approximation with better additive error. 

All known diameter approximation algorithms except for the original Aingworth et al.~\cite{aingw} algorithm and the $\tilde{O}(mn^{2/3})$ time algorithm of \cite{chechikdiam} are {\em randomized}. This raises our first natural question:

\begin{center}
{\em Question 1: Can one obtain deterministic diameter approximation algorithms with the same guarantees as the known randomized algorithms?}
\end{center}

Another important area of research in APSP approximation for undirected graphs is the construction of space-efficient approximate {\em distance oracles} that answer distance queries in constant time. Thorup and Zwick \cite{thorupzwickdos} first provided a randomized construction that for every integer $k\geq 1$ and every undirected $n$-node $m$-edge graph with nonnegative weights preprocessed a distance oracle in $\tilde{O}(mn^{1/k})$ expected time that could provide
a $(2k-1)$-approximation to any distance in $O(k)$ query time. Later, Roditty, Thorup and Zwick \cite{RodittyTZ05det} derandomized the construction. Chechik \cite{chechik2014approximate}, following Wulff-Nilsen \cite{wulff2013approximate}, brought down the query time to a global constant, independent of $k$.

The space usage of these oracles is known to be optimal under the Erd\"os Girth Conjecture. 
Recent work \cite{abboud1,abboud2,ceyinzhan} has focused on showing that the preprocessing time of the distance oracles might be optimal. The latest of these results is that under the $3$SUM Hypothesis, $n^{1+1/(2k-1)-o(1)}$ time is needed to compute any $(2k-\delta)$-distance oracle of an $O(n)$-edge graph that can support $n^{o(1)}$ time queries.  

The literature on distance oracles is vast and there is a huge toolbox of techniques for constructing them. By and large, the known diameter approximation algorithms use a very small subset of the techniques for distance approximation: mainly hitting sets of balls around vertices, combined with Dijkstra's algorithm. Moreover, there are many different constructions of distance oracles (with different guarantees) besides the original Thorup and Zwick construction: e.g. \cite{akavroditty2do2020,beyondTZ14,newinfinity}. Meanwhile, there are only a small handful of diameter approximation algorithms, the main one being the Cairo-Grossi-Rizzi construction (which itself is based on \cite{aingw,rv13}).
Here we ask:

\begin{center}
{\em Question 2: Can we obtain new diameter algorithms by employing the vast techniques from the distance oracle literature? Can the many known distance oracle constructions lead to new diameter algorithms?}
\end{center}

\subsection{Our results}

In this paper we establish a strong connection between distance oracles and diameter approximations, which allows us to address both Question 1 and Question 2. 

We first employ a classic deterministic distance oracle construction to obtain a diameter approximation algorithm matching the tradeoff curve of Cairo, Grossi and Rizzi. Beyond derandomization, this new algorithm is a simpler version of the longstanding state of the art diameter approximation and the first improvement over it in almost a decade. 

We further develop tools to derandomize other known diameter approximations, showing that randomization is not needed in most of the current best known diameter approximation algorithms.
Additionally, we leverage other existing distance oracles to obtain new diameter approximation algorithms.

While this work takes only a small step towards closing the gap between the best known upper and lower bounds for diameter approximations, we hope the powerful connection we demonstrate between distance oracles and diameter approximations will open the door to new advancements in this field using the extensive range of techniques used in approximate distance oracles.





\subsection{Derandomization of Diameter Algorithms}
Our first result is a full derandomization of the Cairo-Grossi-Rizzi (CGR) \cite{cgr} approximation scheme for diameter in undirected graphs. We show how to achieve this best known runtime/approximation tradeoff using the classic $(2k-1)$-approximate distance oracle structure of Thorup and Zwick \cite{thorupzwickdos}, which in turns allows us to use the subsequent derandomization of this distance oracle to derandomize our diameter approximation scheme.

\begin{theorem}[\autoref{thm:detcgrdiam} in the body]\label{thm:cgr}
    Given an undirected graph with nonnegative edge weights bounded by $M$ and an integer $k\geq 2$, one can compute in \emph{deterministic} $\tilde{O}(mn^{1/k})$ time an estimate $\tilde{D}$ of the diameter $D$ satisfying $\frac{2^{k-1}}{2^k - 1}D - \alpha \leq \tilde{D}\leq D$ where $\alpha = \frac{2^{k-1}-1}{2^k-1}M$.
\end{theorem}
We note that our deterministic algorithm runs in the exact same time as the original algorithm of \cite{cgr}, without adding any logarithmic factors to the runtime, which is often the cost of derandomizations. Furthermore, our algorithm also derandomizes the radius and eccentricity approximation schemes of \cite{cgr}. While there is an $\tilde{O}(m)$ expected time $2$-approximation algorithm for all eccentricities \cite{choudhary2020extremal,towardsdiam} which subsumes the CGR eccentricities scheme, it is not known how to derandomize these $2$-approximation algorithms. Thus, our derandomization of CGR presents the best known deterministic approximation algorithms for all nodes eccentricities.


Our next contribution is a derandomization of another tool frequently used in distance oracles: the computation of clusters and balls that are simultaneously small, first introduced by Thorup and Zwick \cite{thorup2001compact}.
For a vertex $v$ and a set of vertices $S$, the $S$-ball of $v$, $B_S(v)$ is the set $\{x\in V~|~d(v,x)<d(v,S)\}$. The $S$-cluster $C_S(v)$ is $\{x\in V ~|~v\in B_S(x)\}$. If the graph is directed, we can define
$C_S^R(v)$ and $B_S^R(v)$ to be $C_S(v)$ and $B_S(v)$ in the graph with edge directions reversed.

Clusters and balls are in a sense inverses of each other. The fast computation of small clusters and balls is at the heart of fast distance oracle preprocessing. Approximate distance oracles often use a random resampling technique to construct a set with worst case guarantees on the size of \emph{all} of its clusters (this technique originates from \cite{thorup2001compact}). Here, we show that the same can be achieved deterministically in the same time complexity. 


\begin{restatable}{theorem}{ThmBallsCluster}\label{thm:ballsclusters}
    Given a weighted, directed or undirected graph $G$ and $\ell >0$, one can compute a set $S\subseteq V$ of size $O(\frac{n}{\ell}\log n)$ deterministically in time $O((m+n \log{n})\ell)$ such that for every $v\in V$, $|C_S(v)|,|B_S(v)|,|C_S^R(v)|,|B_S^R(v)|= O(\ell)$. 
\end{restatable}

\subsection{New Diameter Approximation Algorithms from Known Distance Oracles}
Using this new deterministic tool, together with ideas from two known distance oracle constructions, we obtain new deterministic diameter approximation algorithms.

The first result we address concerns removing the additive error from the original $3/2$-diameter approximation of Aingworth et al. \cite{aingw}. Chechik et al. \cite{chechikdiam} gave a $O(mn^{2/3}\log^{5/3}(n))$ time deterministic $3/2$-approximation to the diameter of any directed or undirected graph, getting rid of the additive error at the cost of a slower runtime.

Using the tools from a $(2,1)$ distance oracle data structure \cite{beyondTZ14}, combined with our deterministic cluster and ball computation tool, we obtain an arguably simpler algorithm that achieves roughly the same result as \cite{chechikdiam}, and even removes a logarithmic factor in the running time. 

\begin{theorem}[\autoref{thm:diam32approx} in the body]
    Given a directed graph with nonnegative edge weights and (unknown) diameter $D$, one can compute in deterministic $O(mn^{2/3}\log^{2/3}n)$ time an estimate $\tilde{D}$ satisfying $\frac{2D}{3}\leq \tilde{D}\leq D$.
\end{theorem}

We then combine the basis of the  $(2+\varepsilon, 5)$ distance oracle of \cite{akavroditty2do2020} with our new deterministic tools to obtain a brand new diameter approximation algorithm, the first such algorithm whose multiplicative approximation is not of the form $2-\frac{1}{2^k}$.

\begin{theorem}[\autoref{thm:diam53approx} in the body]
    Given an unweighted, undirected graph with (unknown) diameter $D$, one can compute in deterministic $O(nm^{3/5}\log^{8/5} n)$ time an estimate $\tilde{D}$ satisfying $\frac{3D}{5} - \alpha\leq \tilde{D}\leq D$ for $\alpha =\max(\frac{6}{5}, \frac{5}{3}-\frac{D}{15}))$.
\end{theorem}

 This $\tilde{O}(nm^{3/5})$ almost-$5/3$ approximation algorithm is faster than all the known almost-$3/2$-diameter approximation algorithms (which run in $\tilde{O}(\min\{n^2,m\sqrt n\})$ time) for every graph density $m$ in the interval $[n^{5/4},n^{5/3}]$. 
 
 Notice that the approximation guarantee $5/3$ is between the $3/2$ and $7/4$ guarantees of the CGR algorithm \cite{cgr} for $k=2$ and $k=3$ respectively, so that this is a brand new point on the trade-off curve for diameter approximation algorithms.
 
Over the last decade, there has been significant progress in fine-grained conditional lower bounds for diameter. 
Starting with the work of \cite{rv13,towardsdiam,bonnetdiamhardness} and culminating in the results of \cite{raylidiamhardness,diamhardnessdir,diamhardnessundir}, we now have a full lower bound trade-off curve for diameter even in undirected and unweighted graphs: under the Strong Exponential Time Hypothesis, for every integer $k\geq 2$, any $2-1/k-\eps$ approximation algorithm needs $mn^{1/(k-1)+o(1)}$ time. Using $k=2,3$, we get that the $3/2$-approximation algorithms of \cite{rv13,chechikdiam} are optimal under SETH, and using $k$ in the limit, we get that the folklore linear time $2$-approximation cannot be improved.



If we believe that the known conditional lower bounds are best possible, we should suspect that (for $k=3,4$ in the trade-off) there should be a $5/3$-approximation algorithm running in $\tilde{O}(m^{4/3})$ time. Unfortunately, the known $\tilde{O}(m^{4/3})$ time approximation algorithm on the CGR curve \cite{cgr} only achieves a $7/4$ approximation. Our new $5/3$-approximation algorithm can be seen as a step towards the desired next algorithmic point on the lower bound curve after $3/2$.


\subsection{Applications and Further Derandomizations}
To illustrate the usefulness of our new techniques, we obtain derandomizations of several known results related to various problems in the field of shortest path computations.


    

First we consider the $<3$ APSP approximations of Baswana and Kavitha \cite{baswanakavitha10}. They construct a $2$-approximation and a $7/3$-approximation algorithm using a hierarchy of randomly sampled sets and a randomized version of \autoref{thm:ballsclusters}. This randomized hierarchy needs to satisfy two conditions. First, that for every set $S$ of size $\tilde{O}(q)$ the set of edges $E_S(v)$ defined as the edges incident to $v$ of weight $<d(v,s)$ satisfies $|E_S(v)|=O(n/q)$.  This can be achieved deterministically by constructing a greedy hitting set to the $n/q$ highest weight incident edges of every vertex, which can be done in $O(n^2)$ time. Second, the hierarchy of sets $\emptyset = A_k \subseteq A_{k-1}\subseteq \ldots, \subseteq A_0 = V$ needs to satisfy that for every vertex $v$, $|B_i(v)|=|\{x\in A_i:d(v,x) < d(v, A_{i+1}\}|\leq n^{1/k}$, which can be achieved using the deterministic distance oracle construction of Roditty, Thorup and Zwick \cite{RodittyTZ05det}. Taking the sets $A_0, \ldots, A_k$ obtained the construction of \cite{RodittyTZ05det}, and taking a union of each $A_i$ with a hitting set to the $n/|A_i|$ highest weight incident edges of every vertex guarantees that both conditions hold. Together with \autoref{thm:ballsclusters}, this gives us a deterministic version of the following result:

\begin{theorem}[Derandomization of Algorithms 9 and 10 in \cite{baswanakavitha10}]
    A given weighted undirected graph on $n$ vertices can compute all-pairs $2$-approximate distance in \emph{deterministic} time $\tilde{O}(m\sqrt{n} + n^2)$ or all pairs $7/3$-approximate distance in deterministic time $\tilde{O}(m^{2/3}n + n^2)$.
\end{theorem}


Using a similar idea, we can derandomize the $\tilde{O}(n^2)$-time construction of a $(2k-1)$-approximate distance oracle of Baswana and Kavitha \cite{baswanakavitha10}. In this algorithm, they use a similar hierarchy of sets $A_{k-1}\subseteq \ldots\subseteq A_0=V$ of sizes $A_i = \tilde{O}(n^{1-i/k})$ which satisfy the same properties. To avoid the $\tilde{O}(mn^{1/k})$ runtime, they authors construct a $3$-spanner on $\min(m, n^{3/2})$ edges and use it to construct a graph on $\bar{m}=\min(m, n^{2-1/k})$ in expected linear time. They proceed to compute $A_{k-1}$-balls using only the edges of this graph and compute the sets $B_i(v)$ for every vertex $v$ in quadratic time.

We can derandomize this algorithm by beginning with sets $ S_{k-i}$ which are deterministic hitting sets to the $n^{i/k}$ highest weight edges adjacent to every vertex.  We can then compute the same $3$-spanner in deterministic linear time \cite{RodittyTZ05det} and compute the $n^{1/k}$ closest vertices to each vertex in the corresponding subgraph in time $\tilde{O}(\bar{m}n^{1/k})=\tilde{O}(n^2)$. We take $A_{k-1}$ to be the union of deterministic hitting set to these neighborhoods and $S_{k-1}$. For subsequent $A_i$'s, we compute $B_{i+1}(v)$ for every vertex in the graph $(V, E_{S_{i}})$. We then take $A_i$ to be the union of $S_i$ and a deterministic hitting set of all $B_{i+1}$. This approach gives us the following deterministic result: 

\begin{theorem}[Derandomization of Theorem 6.6 of \cite{baswanakavitha10}]
    An undirected, weighted graph on $n$ vertices and $m$ edges can be preprocessed in \emph{deterministic} time $\tilde{O}(\min(n^2, kmn^{1/k}))$ to compute a $(2k-1)$-distance oracle of size $\tilde{O}(kn^{1+1/k})$ for any integer $k > 2$.
\end{theorem}

Because \autoref{thm:ballsclusters} is a direct derandomization of Corollary 3.3 from \cite{thorup2001compact}, we also get a deterministic version of their algorithm for constructing a stretch $3$ routing scheme that works in the same time complexity (note that their scheme cites a \textit{deterministic} construction of a 2-level hash table that supports lookups in worst-case constant time). 

Backurs et al. \cite{towardsdiam}\footnote{We note that Theorem 38 of \cite{towardsdiam} (a randomized $O(n^{2.05})$ algorithm that yields a slightly better almost-$3/2$ diameter approximation and almost-$5/3$ eccentricity approximation for undirected, unweighted graphs) can directly be derandomized using the deterministic distance oracle techniques from \cite{RodittyTZ05det}.} use the randomized small clusters technique from \cite{thorup2001compact} to obtain an almost-$3/2$ approximation of undirected unweighted diameter in \textit{expected} $O(n^2\log{n})$ time. The randomness in their algorithm arises from randomly sampling a hitting set $A$ of size $|A| = O(\sqrt{n}\log{n})$ such that $|C_A(w)| = O(\sqrt{n})$ for all $w$. Then, the algorithm iterates over all pairs of vertices within every cluster, resulting in an $O(n \times |C_A(w)|^2) = O(n^2)$ contribution in time complexity. The algorithm also uses a result from Knudsen \cite{additivespanners} to obtain an additive $2$ spanner $H$ with $O(n^{1.5})$ edges in deterministic $O(n^2)$ time, and computes $|A|$ shortest path threes in $H$ in $O(|A| \cdot n^{1.5}) = O(n^2\log{n})$ time.

To derandomize the result, we first compute the nearest $\ell=O(\sqrt{n})$ neighbors to each node, which can been listed out in $O(n\ell^2)$ time (see the construction of a \textit{k-partial-BFS tree} in \cite{doi:10.1137/S0097539797327908}). Then, we use our \autoref{lemma:smallclusters} to compute the hitting set $A$ with the same guarantees in deterministic $O(n\ell\log{n}) = O(n^{1.5}\log{n})$ time, giving us the following result: 

\begin{theorem}[Derandomization of Algorithm 32 in \cite{towardsdiam}]
    There is a deterministic $O(n^2\log{n})$ time algorithm that, for an unweighted undirected graph $G$ with diameter $D = 3h + z$ where $h$ is a positive integer and $z \in [0, 1, 2]$, outputs a value $\hat{D}$ satisfying

\[
\begin{cases}
2h - 1 & \text{if } z \in [0, 1] \\
2h     & \text{if } z = 2
\end{cases}
\quad \leq \hat{D} \leq D
\]
\end{theorem}





\paragraph{Related work.}
Besides the typical definition of diameter, there are several other notions of diameter in directed graphs, using different symmetric notions of distance (\cite{fixedparamsubq}). Several works \cite{fixedparamsubq,chechiklinearmindiam2022,bergermindiam2023,mindistance2019,mindistanceDAG2021} provide approximation algorithms for the so-called min-diameter which is defined as $\max_{u,v}\min\{d(u,v),d(v,u)\}$. The roundtrip diameter is $\max_{u,v} d(u,v)+d(v,u)$, and approximation algorithms and conditional lower bounds for it were studied by \cite{directeddiam,fixedparamsubq}. Other notions such as $ST$ diameter and bichromatic diameter were studied by \cite{towardsdiam,bichromatic}. Most of the known algorithms are randomized. We suspect that at least some of them can be derandomized using our techniques.

There is a lot of work on APSP approximation. Besides the approximation algorithms mentioned so far, there is work on additive approximations for APSP \cite{doi:10.1137/S0097539797327908,deng2022new,durr2023improved,saha2024faster} and multiplicative $2$-approximations \cite{aingw,doi:10.1137/S0097539797327908,CohenZ01,baswanakavitha10,Kavitha12,deng2022new,durr2023improved,Roditty23,dory2024fast,saha2024faster}. We can't hope to capture all of the vast literature on shortest paths approximation. We invite the reader to look at the references in the above papers for more.


%% file: 2_prelim.tex
\section{Preliminaries}
Let $G=(V,E)$ be a directed or undirected graph on $n$ vertices and $m$ edges. When $G$ is a weighted  graph denote by $w(u,v)$ the weight of the edge $(u,v)$. Given a pair of vertices $u,v\in V$ denote by $d_G(u,v)$ the length of the shortest path from $u$ to $v$ in $G$. When $G$ is clear from context we drop the subscript.

Denote by $\epsilon(u)=\max_{v\in V}d(u,v)$ the eccentricity of a vertex $u$ and let $D=\max_{u\in V}\epsilon(u)$ be the diameter of the graph and $R=\min_{u\in V}\epsilon(u)$ be the radius of the graph. 

We say that a value $\tilde{D}$ is an $(\alpha, \beta)$-approximation to a parameter $D$ if $\alpha D + \beta \leq \tilde{D} \leq D$. When $\beta =0$ we have a multiplicative approximation and call it an $\alpha$-approximation. When $\beta >0$ is a constant we refer to $\tilde{D}$ as an `almost' $\alpha$-approximation.

Given a vertex $v$, denote by $N(v)=\{u\in V : (v,u)\in E\}$ its neighborhood. Given an additional set $S$ define the $S$-ball of a vertex $v$ by $B_S(v)=\{u\in V: d(v,u) < d(v,S)\}$. Define the $S$-cluster as $C_S(x)=\{v\in V:x\in B_S(v)\}$. In some case we are interested in the neighborhood of the $S$-ball of a vertex, we denote this set by $B_S^+(x) = \bigcup_{w\in B_S(x)}N(w)$. 

If $G$ is a directed graph, we further define the incoming $S$-balls and clusters $B_S^R(v), C_S^R(v)$ to be the $S$-balls and clusters in the graph with edge direction reversed.




Given a subset  $U\subseteq V$, and integer $q\geq 1$, for $v\in V$ the set $U_q(v)$ is defined to be the $q$ closest vertices of $U$ to $v$, i.e. $|U_q(v)|=q$ and whenever $x\in U_q(v), y\in U\setminus U_q(v)$, then $d(x,v)\leq d(y,v)$. When creating $U_q(v)$, ties are broken according to some permutation of the vertices, e.g. the lexicographic order of their names.

We use the following result concerning the fast, deterministic computation of the sets $U_q$.

\begin{lemma}[Theorem 2 of \cite{RodittyTZ05det}]\label{lm:getneighbors}
For any directed $n$-node, $m$-edge $G=(V,E)$ with positive edge weights, any $U\subseteq V$, $1\leq q\leq |U|$, one can compute $U_q(v)$ for all $v\in V$ in $O((m+n\log n)q)$ time deterministically 
by performing $q$ SSSP\footnote{Single source shortest paths, either BFS in $O(m)$ time in undirected graphs or Dijkstra's algorithm in $O(m+n\log n)$ time in directed graphs. We make this distinction since in some cases we wish to be careful about our log factors.} computations on graphs with $O(n)$ nodes and $O(m)$ edges.
\end{lemma}

Another fundamental tool we use is the construction of a deterministic hitting set:

\begin{lemma}[Greedy hitting set]
\label{lm:determhit}
Given $S_1,\ldots, S_N\subseteq [n]$ such that $N=\textrm{poly}(n)$ and for all $i\in [N]$, $|S_i|\geq L$, one can compute deterministically in $O(NL)$ time a set $H$ of size $O(n/L \log n)$ such that for all $i$, $S_i\cap H\neq \emptyset$.
\end{lemma}


As mentioned, we derandomize a stronger variant of the balls and clusters construction, one that guarantees that the clusters of all vertices are of bounded size. The proof of this theorem can be found in \autoref{sec:detclusters}.

\ThmBallsCluster*


In such a setting, when all vertices have small clusters, we note that in $O(C\cdot \ell)$ time we can compute not only $B_S(v)$ for all vertices but also $B_S^+(v)$, where $C$ represents the time it takes to compute a single graph SSSP search. We can do this by scanning the edges out of every $B_S(v)$ and using the following claim.
\begin{claim}\label{clm:totalsizeofplusballs}
    $\sum_{v\in V}|B_S^+(v)|=O(m\ell)$.
\end{claim}
\begin{proof}
\[
\sum_{v\in V}|B_S^+(v)|\leq\sum_{v\in V}\sum_{x\in B_S(v)}deg(x) = \sum_{x\in V}\sum_{v\in C_S(x)}deg(x)\leq \ell \sum_{x\in V}deg(x) = O(m\ell).
\]
\end{proof}

%% file: 3_deterministic_cgr.tex
\section{Deterministic $(2-1/2^{k-1})$-Diameter Approximation} \label{sec:detcgr}
In this section we derandomize the longstanding best known diameter approximation tradeoff of Cairo, Grossi and Rizzi \cite{cgr}, achieving a deterministic algorithm that matches the multiplicative and additive approximation guarantees of \cite{cgr} without incurring any additional cost on the runtime (as our algorithm matches the runtime of \cite{cgr} down to the logarithmic dependency). To do so, we avoid randomized sampling techniques and reduce the problem of approximating the graph diameter to a series of queries to the classic approximate distance oracle of Thorup and Zwick \cite{thorupzwickdos}. Using the deterministic construction of this distance oracle of Roditty, Thorup and Zwick \cite{RodittyTZ05det}, we obtain the first deterministic diameter approximation tradeoff for undirected graphs. 

Just as the algorithm of Cairo, Grossi and Rizzi also provides a $(2-\frac{1}{2^{k-1}})$-approximation to the graph radius and a $(3-\frac{4}{2^{k-1}+1})$-approximation to all eccentricities, so can our algorithm. We defer the proof of this result to \autoref{sec:radiuseccapprox}.

We note that by using \autoref{lm:determhit} to construct a hitting set of neighborhoods constructed using \autoref{lm:getneighbors}, we can easily derandomize both the $\tilde{O}(m\sqrt{n})$ time almost-$3/2$ approximation of \cite{rv13} and the $\tilde{O}(m^{3/2})$ time genuine-$3/2$ approximation of \cite{chechikdiam} which work in directed graphs. Derandomizing the full CGR construction is not as simple since the original randomized construction uses random samples to hit large node neighborhoods that would be expensive to construct explicitly. We get around this by adapting the ideas behind the Thorup-Zwick distance oracles.


\begin{theorem}\label{thm:detcgrdiam}
    Given an undirected graph with edges weight bounded by $M$ and an integer $k\geq 2$, one can compute in \emph{deterministic} $\tilde{O}(mn^{1/k})$ time an estimate $\tilde{D}$ satisfying $\frac{2^{k-1}}{2^k - 1}D - \alpha \leq \tilde{D}\leq D$ where $\alpha = \frac{2^{k-1}-1}{2^k-1}M$.
\end{theorem}

\begin{proof}
    We begin by recursively constructing the data structure used in the deterministic $(2k-1)$-approximate distance oracle of \cite{RodittyTZ05det}. Begin with $A_0=V$. 
    Now, assuming we have the set $A_i$, let $N_i(v)$ be the $q=\tilde{O}(n^{1/k})$ closest vertices to $v$ in $A_i$, we will set the exact logarithmic dependency of $q$ in our runtime analysis. Using \autoref{lm:getneighbors}, we can construct these sets in $\tilde{O}(mn^{1/k})$ time. Define $A_{i+1}$ to be a greedy hitting set of the sets $N_i(v)$ for all $v$ using \autoref{lm:determhit}. Note that the lemma gives the following property:

    \begin{claim}
        $|A_i|\leq \tilde{O}(n^{1-i/k})$.
    \end{claim}

    For every $i$, run Dijkstra's algorithm from the set $A_{i+1}$ and let $v_i$ be the furthest vertex of $V$ from the set. Define $B_i(v_i) \coloneqq \{x\in A_i : d(v_i,x) < d(v_i, A_{i+1})\}$. By the definition of $A_i$, $|B_i(v_i)|= \tilde{O}(n^{1/k})$. 

    Finally, we run Dijkstra's algorithm from $A_{k-1}$ and from every vertex $y\in B_i(v_i)$ for $0\leq i\leq k-2$ and return the largest distance found. See Algorithm \ref{alg:detcgr} for the full pseudo-code.

    \begin{algorithm}
         \caption{($2-\frac{1}{2^{k-1}}$)-Diameter Approximation}\label{alg:detcgr}
         \begin{algorithmic}[1]
            \item \textbf{Input:} Weighted, undirected graph $G = (V, E)$.
            \item \textbf{Output:} Diameter approximation $D\geq \tilde{D}\geq (2-1/2^{k-1})D - \alpha$.
            \State $A_0\leftarrow V$.
            \For{$i=0,\ldots, k-2$}
                \State $N_i(v)\leftarrow$ the $q=\tilde{O}(n^{1/k})$ closest vertices to $v$ in $A_i$.\label{ln:startfor}
                \State $A_{i+1}\leftarrow$ hitting set for all the sets $N_i(v)$.
                \State $v_i \leftarrow$ furthest vertex in $V$ from $A_{i+1}$.
                \State $B_i(v_i) \leftarrow \{x\in A_i : d(x,v_i) < d(v_i, A_{i+1})\}$.
                \State Run Dijkstra's algorithm from every $y\in B_i(v_i)$.\label{ln:endfor}
            \EndFor
            \State Run Dijkstra's algorithm from every vertex in $A_{k-1}$.
            \State \Return Largest distance computed in one of the Dijkstra searches.
            \end{algorithmic}
    \end{algorithm}

    \paragraph{Correctness:} 
    We want to show that the distance $\tilde{D}$ returned by the algorithm satisfies the desired distance approximation. Clearly $\tilde{D}\leq D$ as it is a true distance in the graph, so we are left to show the lower bound. To do so we define a series of parameters and use them to prove that at every step either we find a pair of vertices of distance $\geq \frac{2^{k-1}}{2^k-1}D - \alpha$ or the vertex $v_i$ is sufficiently far from $A_{i+1}$ and we proceed inductively.
    
    Fix a pair of diameter endpoints, $d(s,t) = D$ and define the values $\varepsilon,\gamma_1\ldots, \gamma_{k-1},\alpha_2,\ldots, \alpha_{k-1}$ as follows. Let $\varepsilon = \frac{2^{k-1}-1}{2^k-1}$,  we want to show that $\tilde{D}\geq (1-\varepsilon)D - \varepsilon M$. For every $r=0,\ldots, k-2$ define $\gamma_{k-1-r}=\frac{1}{2^r}(1-\varepsilon)-(1-2\varepsilon)$. Finally, define for every $i=1,\ldots, k-2,~\alpha_{i+1} = \gamma_{i+1}-\gamma_i$. 
    
    Note that since the $\gamma_i$'s are increasing we always have that $\alpha_i>0$. Furthermore, $\gamma_1 = 1-2\varepsilon=\frac{1}{2^k-1}>0$ so all $\gamma_i > 0$. Next we note the following property relating our chosen parameters to each other and motivating their definition.

    \begin{proposition}\label{prop:paramprops}
        $1-\varepsilon - \alpha_{i+1} + \gamma_i = \varepsilon$.
    \end{proposition}
    \begin{proof}
        \begin{align*}
            \alpha_{i+1} - \gamma_{i} &= \gamma_{i+1}-2\gamma_i = 
            \frac{1}{2^{k-i-2}}(1-\varepsilon)-(1-2\varepsilon) - \frac{2}{2^{k-i-1}}(1-\varepsilon) + 2(1-2\varepsilon)\\
            &= \left(\frac{1}{2^{k-i-2}}-\frac{2}{2^{k-i-1}}\right)(1-\varepsilon) + (1-2\varepsilon) = 1-2\varepsilon.
        \end{align*}
        Thus, $1-\varepsilon - \alpha_{i+1} + \gamma_i = \varepsilon$.
    \end{proof}

    We wish to show inductively that at every step of the algorithm we either find a pair of point far enough apart or claim a lower bound on the distance between $v_i$ and $A_{i+1}$. We begin with the base case.

    \begin{claim}\label{clm:approxbase}
        If there exists a vertex $x\in A_{k-1}$ such that $d(s,x)\leq \varepsilon (D + M)$ then $\tilde{D}\geq d(x,t)\geq (1-\varepsilon)D - \varepsilon M$.
    \end{claim}
    The proof follows from the triangle inequality. Thus, if after running Dijkstra's from $A_{k-1}$ we don't have a good enough approximation then $d(v_{k-2}, A_{k-1}) > \varepsilon(D + M) = \gamma_{k-1} (D+M)$.

    \begin{lemma} \label{lm:approxstep}
        Suppose $d(v_i, A_{i+1}) > \gamma_{i+1} (D+M)$, then after running Dijkstra's from $B_i(v_i)$ we either have an estimate $\tilde{D} \geq (1-\varepsilon)D - \varepsilon M$, or $d(v_{i-1}, A_i) > \gamma_i (D + M)$.
    \end{lemma}

    \begin{proof}
        When we run Dijkstra's from $v_i$, if we do not get the desired approximation we must have that $d(s, v_i) < (1-\varepsilon) D - \varepsilon M$. Since $d(v_i, A_{i+1}) > \gamma_{i+1} (D+M)$, the set $B_i(v_i)$ contains all vertices of $A_i$ at distance $\leq \gamma_{i+1}(D+M)$ from $v_i$.

        Consider the shortest path between $s$ and $v_i$. There must exists a vertex $c$ on this shortest path such that $d(v_i, c) \leq \alpha_{i+1}(D+M)$ and 
        \[
        d(s, c) \leq ((1-\varepsilon)D - \varepsilon M) - \alpha_{i+1}(D+M) + M = (1-\varepsilon - \alpha_{i+1})(D+M).
        \]

        Consider the distance $d(c, A_i)$. If this distance is greater than $\gamma_i (D+ M)$ then we have that $d(v_{i-1}, A_i) > \gamma_i (D+M)$. Otherwise, there exists $q\in A_i$ such that $d(c,q) \leq \gamma_i (D+M)$. So by the triangle inequality, 
        \[
        d(v_i, q) \leq \alpha_{i+1}(D+M) + \gamma_i (D+M) = \gamma_{i+1}(D+M).
        \]
        Therefore, $q\in B_i(v_i)$ and so we have computed Dijkstra's from it. Furthermore, by \autoref{prop:paramprops},
        \begin{align*}
        d(s,q)&\leq d(s,c) + d(c,q)\leq (1-\varepsilon - \alpha_{i+1})(D+M) + \gamma_i (D+M) \\
        &= (1-\varepsilon - \alpha_{i+1} +\gamma_{i})(D+M) = \varepsilon (D+M).
        \end{align*}

        Meaning $d(q, t) \geq (1-\varepsilon)D - \varepsilon M$ and we have computed the desired approximation. We conclude that after running Dijkstra's from $B_i(v_i)$ either $\tilde{D}$ is sufficiently large or $d(v_{i-1}, A_i) > \gamma_i (D+M)$.
    \end{proof}

    By applying \autoref{lm:approxstep} $(k-1)$ times, using \autoref{clm:approxbase} as the base case, we either compute an estimate $\tilde{D}\geq (1-\varepsilon)D - \varepsilon M$ or have that $d(v_0, A_1) > \gamma_1 (D+ M) = (1-2\varepsilon) (D+M)$.

    In this case, $B_0(v_0)$ contains all vertices of $V$ at distance $\leq (1-2\varepsilon)(D+M)$ from $V_0$. Since we run Dijkstra's from $v_0$ we can assume $d(v_0, s) < (1-\varepsilon)D - \varepsilon M$ or we have already found a sufficiently large distance. Thus, there exists a vertex $y$ on the shortest path between $v_0$ and $s$ such that $d(s, y) \leq \varepsilon(D+M)$ and \[
    d(v_0, y) \leq ((1-\varepsilon)D - \varepsilon M) - \varepsilon(D+M) + M = (1-2\varepsilon)(D+M).
    \]

    Therefore $y\in B_0(v_0)$ and $d(y,t) \geq (1-\varepsilon) D - \varepsilon M$ and so running Dijkstra's from all vertices in $B_0(v_0)$ gives us the desired approximation. We conclude that,
    \[
    \tilde{D}\geq (1-\varepsilon)D - \varepsilon M = \frac{2^{k-1}}{2^k-1}D - \frac{2^{k-1}-1}{2^k - 1}M.
    \]

    \paragraph{Runtime:} 
    Using \autoref{lm:getneighbors}, we can compute the sets $N_i(v)$ in total time $O(q\cdot C)$ where $C=O(m+n\log n)$ is the time it takes to run a single graph search. Next, using \autoref{lm:determhit}, computing $A_{i+1}$ takes $O(n\cdot q)$ time and produces a set of size $|A_{i+1}|=O(\frac{|A_i|}{q}\log n) = O(\frac{n}{q^i}\log^i n)$. Finding $v_i$ and computing $B_i(v_i)$ can be done in $O(C)$ time and running Dijkstra's algorithm from every $y\in B_i(v_i)$ takes $O(q\cdot C)$ as $|B_i(v_i)|\leq |N_i(v_i)| = q$.

    Thus, running lines \ref{ln:startfor}-\ref{ln:endfor} takes $O(q\cdot C + n\cdot q) = O(q\cdot C)$ time, for a total runtime of $O(kq\cdot C)$.

    Finally, running Dijkstra's algorithm from all vertices of $A_{k-1}$ takes time $O(|A_{k-1}|\cdot C) = O(\frac{n}{q^{k-1}}\log^{k-1} n\cdot C)$.

    Setting $q=O(\frac{n^{1/k}\log^{(k-1)/k}n}{k^{1/k}})$ gives a final runtime of $O(C\cdot n^{1/k}\cdot \log^{(k-1)/k}n\cdot k^{(k-1)/k}) = \tilde{O}(mn^{1/k})$.
    
\end{proof}

%% file: 4_deterministic_small_clusters.tex
\section{Deterministic Small Clusters}\label{sec:detclusters}

A random hitting set of size $\Theta(\frac{n}{\ell})$ results in an expected average cluster size $\frac{1}{n}\sum_{w \in V}|C_A(w)|$ of $O(\ell)$. A deterministic greedy hitting set of size $\Theta(\frac{n}{\ell}\log{n})$ results in an $O(\ell)$ on the size of every ball, and furthermore the average cluster size is small as a result: $$\forall v: |B_A(v)| \leq \ell \implies \frac{1}{n}\sum_{w \in V}|C_A(w)| = \frac{1}{n}\sum_{v \in V}|B_A(w)| \leq \ell$$

However, sometimes we want the size of \textit{all} clusters to be bounded by $\ell$. Theorem 3.1 from \cite{thorup2001compact} gave a Las Vegas algorithm for doing so that worked in $O(nl\log{n})$ expected time. Here, we present a deterministic algorithm that works in the same running time.

\begin{lemma}\label{lemma:smallclusters}
    There exists a \textbf{deterministic} algorithm running in $O(n\ell\log{n})$ time that, when given an $n \times \ell$ matrix $M^*$ consisting of the closest $\ell$ vertices to every node (in a directed or undirected graph), constructs a set $A^*$ of size $|A^*| = O(\frac{n}{\ell}\log{n})$ such that all clusters and balls with respect to $A^*$ are of size at most $\ell$.
\end{lemma}

To prove lemma \ref{lemma:smallclusters}, we make use of the following "early hitting set" lemma from \cite{RodittyTZ05det}.

\begin{lemma}[Definition 2, Theorem 3 from \cite{RodittyTZ05det}]\label{lemma:earlyhit}
Let $M$ be an $n \times \ell$ matrix whose elements are taken from a finite set $S$ of size $|S|=s$, let $A$ be a set, and let $P \geq 0$ be a penalty. Let $hit(M_i, A)$ be the index of the first element of $M_i$, the $i$-th row of $M$, that belongs to $A$, or $\ell+P$, if no element of $M_i$ belongs to $A$. Let $hit(M, A) = \sum_{i=1}^{n} hit(M_i, A)$ be the \text{hitting sum} of $A$ with respect to $M$.

Then, for every $0 < p < 1$ there is an $O(n\ell)$ time algorithm which finds a set $A \subseteq S$ for which $\frac{n}{p^2s}|A|+hit(M, A) \leq 3n/p + (1-p)^\ell Pn$.
\end{lemma}

\begin{proof}[Proof of \autoref{lemma:smallclusters}:]

The idea is that early hitting sets satisfy the same property that random sets do in the proof of Theorem 3.1 from \cite{thorup2001compact}.
\paragraph{Algorithm:} 

Construct a greedy hitting set (Lemma \ref{lm:determhit}) of size $O(\frac{n}{\ell}\log{n})$ that hits every row of this matrix. Initialize $A^*$ to be this greedy hitting set.

For $i \in [\log{n}]$, we repeat the following:

Compute the balls $B_{A^*}(v)$ for every vertex $v$ by iterating through each row of $M^*$ and finding the first element in each row which is in $A^*$, and compute the clusters $C_{A^*}(w)$ as the inverse of the balls. Let $W_i = \{w \in V\ |\ |C_{A^*}(w)| > \ell\}$, i.e. all clusters that are still “too big”. Now, we construct a new matrix $M$, where for each row in $M^*$, consider only the elements before the first occurrence of $A^*$ that are also members of $W_i$. Pad all rows of $M$ with a dummy element to make them the same length. 

In Lemma \ref{lemma:earlyhit}, plug in $p = \frac{8n}{|W_i|\ell}, s = |W_i|, k = \ell, P = 0$.

Lemma \ref{lemma:earlyhit} returns a set $A$. If the dummy element is not included in $A$, we add it in. We set $A^* \leftarrow A^* \cup A$. 

\paragraph{Correctness:}

Consider the set $A$ returned by Lemma \ref{lemma:earlyhit} after iteration $i$, and the set $A^*$ after $A$ is merged in.

Lemma \ref{lemma:earlyhit} guarantees $\frac{n}{p^2s}|A| + hit(M, A) \leq 3n/p + (1-p)^\ell Pn$. Because $P = 0$ and all terms are non-negative, we have $\frac{n}{p^2s}|A| \leq 3n/p$ and $hit(M, A) \leq 3n/p$. The first inequality yields $|A| \leq ps = \frac{24n}{\ell}$. From the second inequality, we have $hit(M, A) \leq \frac{3|W_i|\ell}{8}$. 

For a row corresponding to vertex $v$, the set of elements that occur before the first hit (which is the entire row, if no hit occurs), is precisely $B_{A^*}(v) \cap W_i$, and so $hit(M, A) = \sum_{v \in V}|B_{A^*}(v)\cap W_i|$. 

We have $\sum_{w \in W_i} |C_{A^*}(w)| = \sum_v |B_{A^*}(w) \cap W_i|  \leq \frac{3|W_i|\ell}{8}$. Recall that $W_{i+1} \subseteq W_i$ are the clusters still satisfying $|C_{A^*}(w)| > \ell$ at the end of the $i$th iteration. So, we have 

\begin{align*}
&\frac{3|W_i|\ell}{8}\geq hit(M, A) = \sum_{v} |B_{A^*}(w) \cap W_i| = \sum_{w \in W_{i}} |C_{A^*}(w)| \geq \sum_{w \in W_{i+1}} |C_{A^*}(w)| > \ell|W_{i+1}|  \\
&\implies |W_{i+1}| < \frac{|W_i|}{2} \\
\end{align*}

$W_0 \subseteq V \implies |W_0| \leq n$, so in $\log{n}$ iterations, $W_{i}$ becomes the empty set, at which point all clusters satisfy $|C_{A^*}(w)| \leq \ell$. 

We add at most $\frac{24n}{\ell}+1$ elements into $A^*$ at every iteration, so the total number of added elements is $|A^*| \leq O(\frac{n}{\ell}\log{n})$. 

\paragraph{Runtime:}

In each iteration, the size of the matrix $M$ is bounded by $n\ell$. Lemma \ref{lemma:earlyhit} constructs the hitting set in linear time in the size of the matrix. The greedy hitting set finds an initial set $A^*$ satisfying $|B_{A^*}(v)| \leq \ell$ for every $v$, so computing all balls and clusters takes time $\sum_{v \in V} |B_{A^*}(v)| = \sum_{w \in V} |C_{A^*}(w)| \leq n\ell$.

There are $\log{n}$ total iterations, for an overall time complexity of $O(n\ell\log{n})$. 

\end{proof}

Finally, we prove \autoref{thm:ballsclusters}.

\begin{proof}[Proof of \autoref{thm:ballsclusters}:]
    It suffices to prove the claim for a weighted, directed graph. First, we use \autoref{lm:getneighbors} to compute $U_q(v)$ for the set $U=V$ and $q = \ell$ in $O((m+n\log{n})\ell)$ time. Then, we apply \autoref{lemma:smallclusters} to compute a set $S_{F}$ of size $O(\frac{n}{\ell}\log{n})$ such that $|C_{S_F}(v)|, |B_{S_F}(v)| \leq \ell$ in $O(n\ell\log{n})$ time.

    Similarly, we compute a set $S_R$ of size $O(\frac{n}{\ell}\log{n})$ such that $|C_{S_R}(v)|, |B_{S_R}(v)| \leq \ell$ in the same running time.

    Finally, we set $S = S_{F} \cup S_R$. Note that balls and clusters are nonincreasing as we add elements to $S$, so $S$ satisfies $|S| = O(\frac{n}{\ell}\log{n})$ and $|C_S(v)|, |B_S(v)|, |C^{R}_S(v)|, |B^{R}_S(v)| = O(\ell)$.

\end{proof}

%% file: 5_additional_diam_approx.tex
\section{New Diameter Approximation Algorithms}

\subsection{Simplified $(3/2,0)$-Diameter Approximation}
We begin by constructing a simple $3/2$-diameter approximation algorithm with no additive error. Our algorithm runs in time $O(mn^{2/3}\log n^{2/3})$, improving by a log factor upon the runtime of the $(3/2,0)$-approximation algorithm of Chechik et al. \cite{chechikdiam} with a vastly simpler algorithm. We derive this algorithm from queries to an adaptation of a $(2,1)$ distance oracle data structure \cite{beyondTZ14} using the classic balls and clusters approach, together with the observation (e.g. \cite{baswanakavitha10}) that iterating over all the \textit{neighbors} of a vertex's ball can be done efficiently, to avoid the original additive error.

\begin{theorem}\label{thm:diam32approx}
    Given a weighted, directed graph with diameter $D$, one can compute in $O(mn^{2/3}\log^{2/3}n)$ time an estimate $\tilde{D}$ satisfying $\frac{2D}{3}\leq \tilde{D}\leq D$.
\end{theorem}

Recall the $3/2$ diameter approximation algorithm of Roditty and Vassilevska W. \cite{rv13}. Take a set $S$ of size $\tilde{O}(n/\ell)$ such that every $v\in V$ has $|B_S(v)|=O(\ell)$. Run Dijkstra's algorithm to and from every vertex in $S$ and let $w$ be the vertex maximizing $d(w,S)$. Finally, run Dijkstra's from $w$ and into every vertex in $B_S(w)$ and return the largest distance found. By setting $\ell = \sqrt{n}$ the algorithm achieves a runtime of $\tilde{O}(m\sqrt{n})$. 

The approximation guarantee comes from the following argument. Consider a pair of diameter endpoints $d(s,t)=D$. If $d(s, S) \leq \frac{D}{3}$, then there exists a point $x\in S$ with $d(x,t) \geq \frac{2D}{3}$, which gives us our desired approximation. Otherwise, $d(w, S) > \frac{D}{3}$ as it maximizes this distance, meaning $B_S(w)$ contains all points of distance $\leq \frac{D}{3}$ from $w$. If $w$ has distance greater than $\frac{2D}{3}$ to $t$ we are done, otherwise pick a vertex on the shortest path between $w$ and $t$ that is within distance $\frac{D}{3}$ from $w$ - meaning it is contained in $B_S(w)$ - and to $t$ - meaning it has distance $\geq \frac{2D}{3}$ from $s$. Thus, running Dijsktra's from all vertices in $B_S(w)$ runs a search out of this vertex and thus guarantees finding a large enough distance. However, such a vertex does not necessarily exist, as on the shortest path between $w,t$ there could be an edge $x\sim y$ such that $d(w,x), d(y,t) < \frac{D}{3}$ while $d(x,t),d(w,y) > \frac{D}{3}$. This causes the algorithm to incur an additive error proportional to the weight of the edge $x\sim y$, which could be arbitrarily large.

To address this additive error, Chechik et al. \cite{chechikdiam}  perform a binary search to have an approximate value of the diameter and introduce various new tools in a relatively involved algorithm running in $O(mn^{2/3}\log^{5/3}n)$ time. 

We instead propose a simple algorithm, hinging on the observation that for any pair of vertices $u,v$ that have $d(u,S),d(S,v)>D/3$, if the distance between them is smaller than our target of $2D/3$, then there exists an edge connecting their $S$-balls to each other. We adapt the $(2,1)$-distance oracle of \patrascu{} and Roditty \cite{beyondTZ14} to compute the distances between such pairs of points. 

\begin{proof}[Proof of \autoref{thm:diam32approx}:]
We are now ready to formally state our algorithm. See Algorithm \ref{alg:diam32approx} for the full pseudo-code. 

Using \autoref{thm:ballsclusters}, construct a set $S$ of size $O(\frac{n\log n}{\ell})$, using a parameter $\ell$ to be set later, such that all ingoing and outgoing $S$-balls and clusters are of size $O(\ell)$. Compute $B_S^+(v)$ for every $v\in V$ and run Dijsktra's algorithm to and from all vertices in $S$. 

Next, we compute the distances obtained from paths of the form $u\rightsquigarrow x \to y \rightsquigarrow v$ where $x\in B_S(u),y\in B^R_S(v)$ and $(x,y)\in E$. Denote by $\hat{d}(u,v)$ the length of the shortest path of this form from $u$ to $v$ and set $\hat{d}(u,v)=\infty$ if no such path exists, note that in general $\hat{d}(u,v)\geq d(u,v)$. For every vertex $u$, we compute all $\hat{d}(u,v)$ by scanning $B_S(u)$. For each $x\in B_S(u)$, each $y\in N(x)$ and each $v\in C_S^R(y)$ we compute $\hat{d}(u,v)\leftarrow \min(\hat{d}(u,v), d(u,x) + w(x,y) + d(y,v))$. In fact, this is an adaptation to the preprocessing step of the $(2,1)$ distance oracle of \patrascu{} and Roditty \cite{beyondTZ14}, which computes the distances between pairs of points whose $S$-balls intersect in the undirected setting. We can think of this step as the preprocessing of an approximate distance oracle where we later obtain the values $\hat{d}(u,v)$ by querying the oracle.

Now, for every vertex $w$ compute the largest value $\hat{d}(w,v)$ for a vertex $v$ such that $d(w,S)\leq d(S,v)$. Denote this value by $\hat{\epsilon}(w)$. Similarly, define $\hat{\epsilon}^R(w)$ to be the largest value $\hat{d}(v,w)$ for a vertex $v$ such that $d(S,w)\leq d(v,S)$. We choose the vertex $w$ that maximizes the expression $\min(3d(w,S), \hat{\epsilon}(w))$ and the vertex $w^R$ that maximizes the expression $\min(3d(S,w^R), \hat{\epsilon}^R(w)^R)$. We claim that running Dijkstra's from $w$ and to $w^R$ will result in the desired approximation.


\begin{algorithm}
     \caption{(3/2,0)-Diameter Approximation}\label{alg:diam32approx}
     \begin{algorithmic}[1]
        \item \textbf{Input:} Weighted, directed graph $G = (V, E)$.
        \item \textbf{Output:} Diameter approximation $D\geq \tilde{D}\geq \frac{2D}{3}$.
        \State $S\leftarrow$ set of size $O(\frac{n}{\ell}\log n)$ such that $\forall v\in V~~|B_S(v)|,|C_S(v)|,|B_S^R(v)|,|C_S^R(v)|=O(\ell)$.
        \State Compute $B_S^+(v),B_S^R(v)~\forall v\in V$ and run Dijkstra's to and from every $x\in S$.\label{ln:searchfromS}
        \For{$u\in V,~x\in B_S(u),~y\in N(x),~v\in C_S^R(y)$}
            \State $\hat{d}(u,v)\leftarrow \min(\hat{d}(u,v), d(u,x)+w(x,y)+d(y,v))$.
        \EndFor
        \For{$w\in V$}
            \State $\hat{\epsilon}(w)\leftarrow \max_{v:d(w,S)\leq d(S,v)}\hat{d}(w,v)$.\label{ln:computeeps}
            \State $\hat{\epsilon}^R(w)\leftarrow \max_{v:d(S,w)\leq d(v,S)}\hat{d}(v,w)$.\label{ln:computeepsrev}
        \EndFor
        \State $w\leftarrow \arg \max_w \min(3d(w,S), \hat{\epsilon}(w))$.\label{ln:selectw}
        \State $w^R\leftarrow \arg \max_w \min(3d(S,w), \hat{\epsilon}^R(w))$.\label{ln:selectwrev}
        \State Run Dijsktra's algorithm from $w$ and to $w^R$. \label{ln:lastdijkstra}
        \State \Return Largest distance computed in one of the Dijkstra searches.
        \end{algorithmic}
\end{algorithm}

\paragraph{Correctness:} Clearly the value $\tilde{D}$ returned by the algorithm satisfies $\tilde{D}\leq D$, as it is a distance in the graph. We are left to show that $\tilde{D}\geq 2D/3$.

Fix a pair of diameter endpoints $d(s,t)=D$. If $d(s,S)\leq \frac{D}{3}$, then there exists a point $x\in S$ such that $d(s,x)\leq \frac{D}{3}$ and so by the triangle inequality $d(x,t)\leq \frac{2D}{3}$. Thus, the algorithm will have found a sufficiently large distance in line \ref{ln:searchfromS}. Similarly, if $d(S,t)\leq \frac{D}{3}$ we are done, so we can assume that $d(s,S), d(S,t)>\frac{D}{3}$.

W.l.o.g we can assume $d(s,S)\geq d(S,t)$, as the reverse case is symmetric. Thus we have that $\hat{\epsilon}(s)\geq \hat{d}(s,t)\geq d(s,t)=D$. Since $d(s,S)>\frac{D}{3}$, we conclude that line \ref{ln:selectw} selects a $w$ such that,
\[
\min (3d(w,S), \hat{\epsilon}(w)) \geq \min (3d(s,S), \hat{\epsilon}(s)) \geq D.
\]
Therefore, $d(w, S)\geq \frac{D}{3}$ and $\hat{\epsilon}(w)\geq D$, meaning there exists a vertex $v$ with $d(S,v)\geq d(w,S)>\frac{D}{3}$ such that $\hat{d}(w, v) \geq D$. We claim that $d(w, v) \geq \frac{2D}{3}$ and so running Dijkstra's from $w$ will obtain the desired approximation.

Indeed, if $d(w, v) < \frac{2D}{3}$ then there exists an edge on the shortest path between $w$ and $v$, $(x,y)\in E$ such that $d(w, x) < \frac{D}{3}$ and $d(y, v)< \frac{D}{3}$. Since $d(v, S)\geq d(w, S)\geq \frac{D}{3}$ we have that $x\in B_S(w)$ and $y\in B_S^R(v)$. Therefore the path $w\rightsquigarrow x \to y \rightsquigarrow v$ would have been considered in the computation of $\hat{d}(w, v)$ and we would have $\hat{d}(w, v)\leq \frac{2D}{3}$, contradiction. We conclude that $d(w,v)\geq \frac{2D}{3}$ and the largest distance found from running Dijkstra's out of $w$ obtains the desired approximation.

\paragraph{Runtime:} Computing the set $S$, as well as all $B_S^+,B_S^R$ balls and the distances within them, takes $O(m\ell \log n)$ time. Running Dijsktra's to and from every vertex in $S$ takes $O(\frac{mn}{\ell}\log n)$ assuming $m = \Omega(n\log n)$. If this is not the case and $m=o(n\log n)$ we can compute a $(3/2,0)$-approximation in $\tilde{O}(m^{3/2})\leq O(mn^{2/3})$ time \cite{chechikdiam}.

Computing the values $\hat{d}(u,v)$ takes time:
\[
\sum_{v\in V}\sum_{y\in B_S^+(v)}|C_S(y)|\leq \sum_{v\in V}|B_S^+(v)|\cdot O(\ell)\leq O(m\ell^2).
\]

The values $\hat{\epsilon}(v),\hat{\epsilon}^R(v)$ can be computed at the same time as the $\hat{d}$ values by sorting the vertices in order of distance to/from $S$ ahead of time and only considering the $\hat{d}$ values computed. If not all relevant vertices $v$ have a value $\hat{d}(w,v)$ or $\hat{d}(v,w)$ computed, we assign the value of $\infty$. Finally, lines \ref{ln:selectw}, \ref{ln:selectwrev} and \ref{ln:lastdijkstra} run in near linear time. 

The total runtime of the algorithm therefore comes out to $O(m\ell\log n + \frac{mn}{\ell}\log n + m\ell^2)$. Setting $\ell = (n\log n)^{1/3}$ we get a runtime of $O(mn^{2/3}\log^{2/3}n)$.

\end{proof}

\subsection{New $(5/3,7/3)$-Diameter Approximation}
In the next section we provide the first diameter approximation algorithm with multiplicative error not of the form $2-\frac{1}{2^k}$. Our algorithm runs in time $\tilde{O}(nm^{3/5})$ and obtains a multiplicative error of $5/3$. In the worst case, the algorithm has an additive error of $+7/3$. This algorithm is faster than the current best $3/2$-approximation algorithms (with additive errors) in the regime where $n^{5/4}\leq m \leq n^{5/3}$.

The algorithm uses the structure of Akav and Roditty's $(2+\varepsilon, 5)$-distance oracle \cite{akavroditty2do2020}, together with the balls and clusters techniques of classic diameter approximation algorithms. As we cannot afford to preprocess the entire $(2+\varepsilon, 5)$-approximate distance oracle, we construct a partial distance oracle treating only pairs of points with one endpoint in a particular set. While the original work is randomized, we use a deterministic hitting set construction while building the distance oracle to achieve a deterministic diameter approximation algorithm.

\begin{theorem}\label{thm:diam53approx}
    Given an unweighted, undirected graph with diameter $D$, one can compute in $O(nm^{3/5}\log^{8/5} n)$ time an estimate $\tilde{D}$ satisfying $\frac{3D}{5} - \alpha\leq \tilde{D}\leq D$ for $\alpha =\max(\frac{6}{5}, \frac{5}{3}-\frac{D}{15}))$.
\end{theorem}

The above algorithm gives us a value $\tilde{D}$ such that $D\leq \frac{5\tilde{D}}{3}+ \frac{5\alpha}{3}\leq D + \frac{5\alpha}{3}$, i.e. a $\left(\frac{5}{3}, \max(2, \frac{25-D}{9})\right)$-approximation to the diameter.
If the diameter of the graph is large enough, $D\geq 7$, we have $\alpha = \frac{6}{5}$ and obtain a $(5/3, 2)$ approximation. In the worst case we have $D=4$, as $D=1$ or $D=2$ can be checked in linear time and if $D=3$ a single BFS finds a distance of $\geq 2=\frac{2D}{3}$. In this case we get $\alpha = \frac{7}{5}$ 
and obtain a $(5/3,7/3)$-approximation.\\

The idea of our algorithm is as follows. We construct sets $S_1$, $S_2$ of size $\tilde{O}(n/\ell)$ and $\tilde{O}(\ell)$ respectively. We can afford to run BFS from all the points in $S_1$ and from a ball $B_{S_2}(w)$ of some vertex $w$ in time $\tilde{O}(\frac{mn}{\ell})$. Using the previous ideas for diameter approximations, this gives us the desired approximation if a diameter endpoint is within distance $\frac{2D}{5}$ of $S_1$ or further than $\frac{D}{5}$ away from $S_2$.

For the remaining case, we show that we have a pair of points $x\in S_2$ and $t\in V$ such that both $d(x,t)$ and $d(x, S_1) + d(t,S_1)$ are large. We now construct a restricted, deterministic version of Akav and Roditty's $(2+\varepsilon, 5)$-approximate distance oracle to handle such pairs. The idea of this distance oracle is to construct a spanner $H$ that guarantees an additive 2 approximation to paths that are fully contained within a single $S_1$-ball, and thus a +4 approximations to distances between points whose $S_1$-balls intersect. We now compute all distances out of $S_2$ using the edges of the spanner $H$, this guarantees that the distance obtained, $d_H(u,v)$ is a good approximation for $d(u,v)$ whenever the $S_1$ balls of $u$ and $v$ intersect. If they don't intersect, we can approximate the distance between $u$ and $v$ with $d(u,S_1) + d(v, S_1)$. The distance oracle returns the minimum of these two estimates, to guarantee a lower bound to the true distance between $u$ and $v$. We show for the pair of points $x$ and $t$ mentioned above that this estimate provides the desired approximation.

The construction of the spanner $H$ is also randomized in the original work of \cite{akavroditty2do2020}. The idea of the spanner is to take a hitting set $T$ of size $\tilde{O}(\frac{n}{L})$ that hits the neighborhood of all vertices of degree $>L$. We then compute $B_{S_1}(v)$ for every $v\in T$ and add a shortest path tree spanning  $B_{S_1}(v)$ to $H$. We complete $H$ by adding all edges adjacent to vertices of degree $\leq L$. 

We can derandomize this construction by carefully constructing the sets $T,S_1, S_2$ in a particular order. We first construct $T$ as a deterministic hitting set to all neighborhoods of vertices of degree $>L$. We then compute the $\ell$-nearest neighbors of each vertex in $T$ and construct $S_1$ to be a hitting set for these neighborhoods, to guarantee that all vertices in $T$ have small $S_1$-balls. Finally, for $S_2$ we can afford to use \autoref{thm:ballsclusters} and construct a set such that all vertices have small $S_2$-balls.

\begin{proof}[Proof of \autoref{thm:diam53approx}:] We can now formally state our algorithm. For full pseudo-code, see Algorithm \ref{alg:diam53approx}.

    Let $\ell,L$ be parameters to be set later. Using \autoref{lm:determhit}, construct a set $T$ of size $O(\frac{n}{L}\log n)$ in time $O(nL)$ that hits the neighborhood of every vertex of degree $> L$. For every $u\in T$ compute $W(u)$ to be the $\ell$ closest nodes to $u$. Now, again using \autoref{lm:determhit}, construct a hitting set $S_1$ of size $O(\frac{n}{\ell}\log n)$ in time $O(n\ell)$ to hit $W(u)$ for every $u\in T$. Note that in this case, for every $u\in T$ we have that $|B_{S_1}(u)| \leq |W(u)|= O(\ell)$. Using \autoref{thm:ballsclusters}, construct a set $S_2$ of size $O(\ell\log n)$ in time $O(\frac{mn}{\ell} \log n)$ such that all vertices have $|B_{S_2}(v)|= O(\frac{n}{\ell})$.

    Following the approach of previous diameter approximation algorithms, begin by running BFS from all vertices in $S_1$ and from the set $S_2$. Let $w$ be the furthest point from $S_2$ and run BFS from every point in $B_{S_2}(w)$. 
    
    Fix a pair of diameter endpoints $d(s,t)=D$. If $d(t, S_1)\leq \frac{2D}{5} +\alpha$, then there exists a point $x\in S_1$ such that $d(x,t)\leq \frac{2D}{5}+\alpha$ and so $d(x,s)\geq \frac{3D}{5}-\alpha$ and running BFS from $S_1$ achieves the desired approximation. Similarly, if $d(s, S_1) \leq \frac{2D}{5} + \alpha$ we are done.
    
    Continuing along the analysis of previous diameter approximations, if $d(s,S_2) > \frac{D}{5} + 1- 2\alpha$, then $d(w,S_2) > \frac{D}{5} + 1- 2\alpha$. If $d(w,s) \leq \frac{3D}{5}-\alpha$, then there exists a vertex $y$ on the shortest path between $w$ and $s$ such that $d(w,y) \leq \frac{D}{5} + 1-2\alpha$ and $d(y,s) \leq \frac{2D}{5} + \alpha$. Therefore $y\in B_{S_2}(w)$ and $d(y,t) \geq \frac{3D}{5}-\alpha$, thus running BFS from all points in $B_{S_2}(w)$ obtains the desired approximation.

    We are left to handle the case when $d(s, S_1),d(t, S_1) > \frac{2D}{5} + \alpha$ and $d(s, S_2) \leq \frac{D}{5} + 1 - 2\alpha$. In this case, there exists a vertex $x\in S_2$ such that $d(s,x) \leq \frac{D}{5} + 1 - 2\alpha$, and so $d(x,t) \geq \frac{4D}{5} - 1 + 2\alpha$. Furthermore, since $d(s, S_1) > \frac{2D}{5} + \alpha$, by the triangle inequality we know that $d(x, S_1) \geq d(s,S_1)-d(x,s)>\frac{D}{5}-1 + 3\alpha$.\\

    To handle this final case we construct a restricted, deterministic version of Akav and Roditty's $(2+\varepsilon, 5)$-approximate distance oracle to compute an approximation for distances from $S_2$. We will simply compute an estimate to these distances, but we can think of this step as constructing a distance oracle and querying $\tilde{d}(u,v)$ for every $u\in S_2, v\in V$.

    Construct the following spanner $H$. Initialize the edges of $H$ to contain all edges adjacent to vertices of degree $\leq L$. For every $u\in T$, add a spanning tree of $B_{S_1}(u)$ rooted at $u$ to $H$. As each $S_1$-ball of a vertex in $T$ is of size $O(\ell)$ this results in $H$ having $O(L + \frac{n\ell}{L}\log n)$ edges. Run BFS from every $x\in S_2$ in the spanner $H$ to compute $d_H(x,v)$ for every $x\in S_2, v\in V$. Using a series of claims we will now show that the following distance approximation is a lower bound to the true distance between a pair of points.
    \begin{lemma}\label{lm:distanceestimate}
        \[
        \tilde{d}(u,v) \coloneqq \min(d_{H}(u,v)-4, d(u, S_1) + d(v,S_1)-5)\leq d(u,v).
        \]
    \end{lemma}

    We prove this lemma by considering two cases. Note that we only use $B_{S_1}(u),B_{S_1}(v)$ for analysis and don't have to compute these sets at any point of the algorithm.
    \begin{claim}
        If no shortest path between $u,v$ contains five vertices in $B_{S_1}(u)\cap B_{S_1}(v)$ then $d(u,S_1) + d(v,S_1) \leq d(u,v) + 5$.
    \end{claim}
    \begin{proof}
        Consider a shortest path between $u,v$. Label the vertices $u=x_0, x_1, \ldots,x_k=v$. Let $x_i$ be the furthest vertex on the path from $u$ that is still in $B_{S_1}(u)$ and let $x_j$ be the furthest vertex from $v$ that is still in $B_{S_1}(v)$. Then since we cannot have five vertices on the path in $B_{S_1}(u)\cap B_{S_1}(v)$ we have $i-j < 4$. Since $x_{i+1}\notin B_{S_1}(u)$ we know $d(u,S_1)=i+1$ and similarly $d(v,S_1)=d(u,v) - j + 1$.

        Therefore,
        \[
        d(u,S_1) + d(v,S_1) = i + 1 + d(u,v) - j + 1 = d(u,v) + 2 + i-j \leq d(u,v) + 5.
        \]
    \end{proof}

    \begin{claim}
        If a shortest path between $u,v$ contains five vertices in $B_{S_1}(u)\cap B_{S_1}(v)$ then $d_H(u,v) \leq d(u,v) + 4$.
    \end{claim}
    \begin{proof}
        Let $x_0=u, x_1,\ldots, x_r=v$ be a shortest path between $u,v$ such that $x_{k-2},x_{k-1}, x_k, x_{k+1}, x_{k+2} \in B_{S_1}(u)\cap B_{S_1}(v)$. Consider the path between $u$ and $x_k$. If all vertices on this path are of degree $\leq L$ then the path is contained in $H$ and $d_H(u,x)=d(u,x)$.

        Otherwise, let $x_i$ be the first vertex of degree $>L$ on the path from $u$ to $x$ and let $p\in T$ be a neighbor of $x_i$. We claim that $x_{k}\in B_{S_1}(p)$. Indeed, by the triangle inequality, $d(p, S_1) \geq d(u, S_1) - d(u,p)$. Since $d(u,S_1) \geq k+3$ and $d(u,p)\leq i+1$ we have that $d(p,S_1) \geq k-1+2$. On the other hand, $d(p, x_k) \leq d(x_i, x_k) + 1 = k-i + 1$, and so $x_k \in B_{S_1}(p)$. 

        Therefore, $H$ contains all edges on the path $u\rightsquigarrow x_i$ as they are adjacent to vertices of degree $\leq L$. $H$ contains the edge $x_i\sim p$ as it is part of the spanning tree of $B_{S_1}(p)$. Finally, $H$ contains a path of length $d(p,x_k)$ between $p$ and $x_k$ since $x_k\in B_{S_1}(p)$. We conclude that,
        \[
        d_H(u,x_k) \leq d(u,p) + d(p,x_k) = d(u,x_k) + 2.
        \]

        Similarly, we can show that $d_H(x_k,v) \leq d(x_k,v) + 2$ and so we can conclude that: \[
        d_H(u,v) \leq d_H(u,x_k) + d_H(x_k,v) \leq d(u,x_k) + 2 + d(x_k,v) + 2 = d(u,v) + 4.
        \]
    \end{proof}

    Therefore, the estimate $\tilde{d}(u,v)$ gives us a lower bound to $d(u,v)$. Now consider $\tilde{d}(x,t)$ for the vertex $x\in S_2$ and the diameter endpoint $t$ discussed above. Since $\alpha \geq \frac{5}{3}-\frac{D}{15}$,
    \[
    d_H(x,t) - 4 \geq d(x,t)\geq \frac{4D}{5} - 1 + 2\alpha - 4 \geq \frac{3D}{5} -\alpha.
    \]
    Furthermore, since $\alpha \geq \frac{6}{5}$,
    \[
    d(x,S_1) + d(t,S_1) - 5 \geq \frac{D}{5} - 1 + 3\alpha + \frac{2D}{5}+\alpha - 5 \geq \frac{3D}{5}-\alpha.
    \]

    Therefore, $\tilde{d}(x,t)\geq \frac{3D}{5}-\alpha$. So returning the largest $\tilde{d}$ distance estimate computed will obtain the desired approximation.

    \begin{algorithm}
         \caption{$(5/3,7/3)$-Diameter Approximation}\label{alg:diam53approx}
         \begin{algorithmic}[1]
            \item \textbf{Input:} Unweighted, undirected graph $G = (V, E)$.
            \item \textbf{Output:} Diameter approximation $\tilde{D}$ such that $\frac{3D}{5}-\alpha \leq \tilde{D}\leq D$.
            \State $T\leftarrow$ hitting set for all neighborhoods of degree $> L$ vertices.
            \State Compute the $\ell$-nearest neighbors of every vertex in $T$.
            \State $S_1\leftarrow$ hitting set for the $\ell$-nearest neighborhoods of $T$.
            \State $S_2\leftarrow$ set of size $O(\ell\log n)$ such that all vertices have $|B_{S_2}(v)|=O(\frac{n}{\ell})$.
            \State BFS from all vertices in $S_1$ and from the set $S_2$. \label{ln:bfss1}
            \State $w\leftarrow$ furthest vertex from $S_2$, run BFS from all vertices in $B_{S_2}(w)$. \label{ln:bfsballs2}
            \State Initialize a spanner $H$ to contain edges adjacent to vertices of degree $\leq L$.
            \For{$u\in T$} 
                \State Add a spanning shortest paths tree of $B_{S_1}(u)$ to $H$.
            \EndFor
            \For{$x\in S_2$}
                \State Run BFS from $x$ in $H$ to compute $d_H(x,v)~\forall v\in V$.
                \State $\forall v\in V~~\tilde{d}(x,v) \leftarrow \min(d_H(x,v) - 4, d(x, S_1) + d(v, S_1)-5)$.
            \EndFor
            \State \Return Largest distance or distance estimate $\tilde{d}$ computed.
        \end{algorithmic}
    \end{algorithm}

    \paragraph{Correctness:} By \autoref{lm:distanceestimate}, all values $\tilde{d}$ computed are smaller than true distances in the graph. Therefore the value $\tilde{D}$ returned is bounded above by a true distance and so $\tilde{D}\leq D$.

    As discussed above, if $d(s,S_1)\leq \frac{2D}{5}+\alpha$ or $d(t,S_1)\leq \frac{2D}{5}+\alpha$, we find a sufficiently large distance in line \ref{ln:bfss1}. Furthermore, if $d(s,S_2)\geq \frac{D}{5}+1-2\alpha$ we find a sufficiently large distance in line \ref{ln:bfsballs2}.
    Otherwise, we have a point $x\in S_2$ and $t\in V$ such that $\tilde{d}(x,t)\geq \frac{3D}{5}-\alpha$ and thus the largest $\tilde{d}$ distance computed gives our desired approximation.

    \paragraph{Runtime:} Using \autoref{lm:determhit} to construct $T$ and $S_1$ takes $\tilde{O}(nL + n\ell)$ time, which will be dominated by $O(nm^{3/5}\log ^{8/5}n)$ for the $\ell,L$ we set. Constructing $S_2$ using a greedy hitting set takes $O(\frac{mn}{\ell}\log n)$ time.

    Computing the $\ell$-nearest neighbors of all vertices in $T$ takes $O(|T|\cdot \ell^2 )= O(\frac{n\ell^2}{L}\log n)$ time. Running lines \ref{ln:bfss1} and \ref{ln:bfsballs2} take $O(\frac{mn}{\ell}\log n)$ time.

    The spanner $H$ contains $O(nL)$ edges adjacent to low degree vertices. Adding the spanning trees to the spanner adds $O(\frac{n\ell}{L}\log n)$ edges for a total of $O(nL + \frac{n\ell}{L}\log n)$ edges. Therefore, computing the distances $\tilde{d}$ from $O(\ell \log n)$ vertices takes $O(nL\ell\log^2 n + \frac{n\ell^2}{L}\log^3 n)$.

    This gives us a final runtime of \[
    O\left(\frac{mn}{\ell}\log n + nL\ell\log^2 n + \frac{n\ell^2}{L}\log^3 n\right).
    \]
    Setting $\ell = \frac{m^{2/5}}{\log^{3/5} n}, L = m^{1/5}\log^{1/5}n$ we obtain our desired runtime of $O(nm^{3/5}\log^{8/5} n)$.
\end{proof}




%% file: 6_deterministic_cgr_radius_ecc.tex
\section{Deterministic $(2-\frac{1}{2^{k-1}})$-Radius Approximation and $(3-\frac{4}{2^{k}+1})$-All Nodes Eccentricity Approximation} \label{sec:radiuseccapprox}

In this section we show that the algorithm introduced in the proof of \autoref{thm:detcgrdiam} can also approximate the radius and all eccentricities of a given graph. 

While there is an $\tilde{O}(m)$ expected time $2$-approximation algorithm for all eccentricities \cite{choudhary2020extremal,towardsdiam}, it is not known how to derandomize these algorithms. Thus, the deterministic algorithm below presents the best known deterministic approximation algorithms for all nodes eccentricities.

Consider Algorithm \autoref{alg:detcgrradiusecc}, which runs the same graph searches as Algorithm \ref{alg:detcgr} and outputs radius and eccentricities approximations.

\begin{algorithm}
         \caption{($2-\frac{1}{2^{k-1}}$)-Radius Approximation and $(3-\frac{4}{2^{k-1}+1})$-Eccentricities Approximation}\label{alg:detcgrradiusecc}
         \begin{algorithmic}[1]
            \item \textbf{Input:} Weighted, undirected graph $G = (V, E)$.
            \item \textbf{Output:} Radius approximation $\tilde{R}$ and eccentricity approximations $\tilde{\epsilon}_w$ for all $w\in V$.
            \State $A_0\leftarrow V$.
            \For{$i=0,\ldots, k-2$}
                \State $N_i(v)\leftarrow$ the $q=\tilde{O}(n^{1/k})$ closest vertices to $v$ in $A_i$.\label{ln:startfor}
                \State $A_{i+1}\leftarrow$ hitting set for all the sets $N_i(v)$.
                \State $v_i \leftarrow$ furthest vertex in $V$ from $A_{i+1}$.
                \State $B_i(v_i) \leftarrow \{x\in A_i : d(x,v_i) < d(v_i, A_{i+1})\}$.
                \State Run Dijkstra's algorithm from every $y\in B_i(v_i)$.\label{ln:endfor}
            \EndFor
            \State Run Dijkstra's algorithm from every vertex in $A_{k-1}$.
            \State \Return $\tilde{R}\leftarrow \min_{x\in B_0(v_0)\cup \ldots\cup B_{k-2}(v_{k-2})\cup A_{k-1}}\epsilon(x)$.
            \State \Return $\tilde{\epsilon}_w\leftarrow \max(\max_{y\in \{v_0,\ldots, v_{k-2}\}\cup A_{k-1}}d(w, y), \max_{x\in B_0(v_0)\cup \ldots\cup B_{k-2}(v_{k-2}) }\epsilon(x) - d(x,w))$
            \end{algorithmic}
    \end{algorithm}

    We claim that Algorithm \ref{alg:detcgrradiusecc} gives the following result.

    \begin{theorem}\label{thm:detcgrradiusecc}
        Given an undirected graph with edges weight bounded by $M$ and an integer $k\geq 2$, one can compute in \emph{deterministic} $\tilde{O}(mn^{1/k})$ time an estimate $\tilde{R}$ satisfying $R\leq \tilde{R}\leq \frac{2^k-1}{2^{k-1}}R + \frac{2^{k-1}-1}{2^{k-1}}M$, and estimates $\tilde{\epsilon}_w$ for every $w\in V$ satisfying $\frac{2^{k-1}+1}{3\cdot 2^{k-1}-1}\epsilon(w)-\frac{2^k-2}{3\cdot 2^{k-1}-1}M\leq \tilde{\epsilon}_w\leq \epsilon(w)$.
    \end{theorem}

    We showed the runtime of the algorithm in the proof of \autoref{thm:detcgrdiam}, so we are left to show that approximation guarantees, which we do in the subsequent \autoref{lm:cgrradiusapprox} and \autoref{lm:cgreccapprox}. We will follow a very similar proof idea, adjusting the parameters for each problem.

    \begin{lemma}\label{lm:cgrradiusapprox}
        The value $\tilde{R}$ outputted by algorithm \ref{alg:detcgrradiusecc} satisfies $R\leq \tilde{R}\leq \frac{2^k-1}{2^{k-1}}R + \frac{2^{k-1}-1}{2^{k-1}}M$.
    \end{lemma}

    \begin{proof}
        Since $\tilde{R}$ is the eccentricity of a vertex in the graph we have that $\tilde{R}\leq R$. We are left to show the upper bound.

        Let $c$ be a center of the graph, $\epsilon(c)=R$ and define the values $\delta,\gamma_1, \ldots, \gamma_{k-1}, \alpha_2, \ldots, \alpha_{k-1}$ as follows. Let $\delta = \frac{2^{k-1}-1}{2^{k-1}}$, we want to show that $\tilde{R}\leq (1+\delta)R + \delta M$. For every $s=0,\leq, k-2$ define $\gamma_{k-1-s}=\frac{1}{2^s}+\delta - 1$. Finally, define for every $i=0,1, \ldots, k-2$, $\alpha_{i+1}=\gamma_{i+1}-\gamma_i$. Since $\gamma_i$ are increasing we always have that $\alpha_i>0$ and since $\gamma_1 = 1-\delta = \frac{1}{2^{k-1}}>0$ we have that all $\gamma_i>0$.

        \begin{proposition}\label{prop:radiusparams}
            $1-\alpha_{i+1}+\gamma_i = \delta$.
        \end{proposition}
        \begin{proof}
            \begin{align*}
                \alpha_{i+1}-\gamma_i = \gamma_{i+1}-2\gamma_i = \frac{1}{2^{k-i-2}}+\delta-1 - \frac{2}{2^{k-i-1}}-2\delta+2 = 1-\delta.
            \end{align*}
        \end{proof}

        \begin{claim}\label{clm:radiusbase}
            If there exists a vertex $x\in A_{k-1}$ such that $d(c,x) \leq \delta(R+M)$ then $\epsilon(x) \leq (1+\delta)R + \delta M$.
        \end{claim}
        The proof follows from the triangle inequality, as all vertices are within distance $R$ of $c$. Thus, if after running Dijkstra's from $A_{k-1}$ we don't have a good enough approximation, then $d(v_{k-2}, A_{k-1}) > \delta(R+M) = \gamma_{k-1}(D+M)$.

        \begin{claim} \label{clm:radiusstep}
            Suppose $d(v_i, A_{i+1}) > \gamma_{i+1}(R + M)$, then after running Dijkstra's from $B_i(v_i)$ we either have an estimate $\tilde{R}\leq (1+\delta)R + \delta M$ or $d(v_{i-1}, A_i) > \gamma_i (R + M)$.
        \end{claim}
        \begin{proof}
            Consider the shortest path of length $\leq R$ between $v_i$ and $c$. Let $x$ be a vertex on this path such that $d(v_i, x) \leq \alpha_{i+1}(R+M)$ and $d(x, c) \leq R - \alpha_{i+1}(R+M) + M = (1-\alpha_{i+1})(R + M)$.

            If $d(x, A_i) > \gamma_i(R+M)$ then we have that $d(v_{i-1}, A_i) > \gamma_i (R+M)$. Otherwise, there exists a vertex $q\in A_i$ such that $d(x,q)\leq \gamma_i(R+M)$ and so by the triangle inequality $d(v_i, q) \leq (\alpha_{i+1} + \gamma_i)(R+M) = \gamma_{i+1}(R+M)$. Therefore, $q\in B_i(v_i)$ and so we have computed Dijkstra's from it. Furthermore, by \autoref{prop:radiusparams},
            \[
            d(c,q) \leq (1-\alpha_{i+1} + \gamma_i)(R+M) = \delta(R+M).
            \]
            Therefore, $\epsilon(q)\leq (1+\delta)R + \delta M$ and so we will have found the desired approximation.
        \end{proof}

        By applying \autoref{clm:radiusstep} $(k-1)$ times and using \autoref{clm:radiusbase} as the base case, we conclude that we either compute an estimate $\tilde{R}\leq (1+\delta)R + \delta M$ or have that $d(v_0, A_1) > \gamma_1 (R+M) = (1-\delta)(R+M)$.

        Now consider the shortest path between $c$ and $v_0$. There exists a vertex $x$ on this path such that $d(v_0, x) \leq (1-\delta)(R+M)$ and $d(x,c) \leq R - (1-\delta)(R+M) + M = \delta(R+M)$. Thus, $x\in B_0(v_0)$ and $\epsilon(x) \leq (1+\delta)R + \delta M$ and so we obtain the desired approximation.
    \end{proof}
    
    \begin{lemma}\label{lm:cgreccapprox}
        For every $w\in V$, the value $\tilde{\epsilon}_w$ outputted by algorithm \ref{alg:detcgrradiusecc} satisfies $\frac{2^{k-1}+1}{3\cdot 2^{k-1}-1}\epsilon(w)-\frac{2^k-2}{3\cdot 2^{k-1}-1}M \leq \tilde{\epsilon}_w\leq \epsilon(w)$.
    \end{lemma}

    \begin{proof}
        By the triangle inequality, for any vertex $x$, $\epsilon(x) \leq \epsilon(w)+d(w,x)$, therefore $\epsilon(x)-d(w,x)$ always gives a lower bound to the eccentricity of $w$. Thus, the estimate $\tilde{\epsilon}_w$ computed for each vertex satisfies $\tilde{\epsilon}_w\leq \epsilon(w)$. We are left to show the lower bound.

        Let $\hat{w}$ be the furthest vertex from $w$ in the graph, $d(w,\hat{w})= \epsilon(w)$,  and define the values $\delta,\gamma_1, \ldots, \gamma_{k-1}, \alpha_2, \ldots, \alpha_{k-1}$ as follows.  Let $\delta = \frac{2^k-2}{3\cdot 2^{k-1}-1}$, we want to show that $\tilde{\epsilon}_w\geq (1-\delta)\epsilon(w) - \delta M$. For every $s=0,\leq, k-2$ define $\gamma_{k-1-s}=\frac{1}{2^s}(1-\frac{\delta}{2})+\frac{3\delta}{2} - 1$. Finally, define for every $i=0,1, \ldots, k-2$, $\alpha_{i+1}=\gamma_{i+1}-\gamma_i$. Since $\gamma_i$ are increasing we always have that $\alpha_i>0$ and since $\gamma_1 = 1-\frac{3\delta}{2} = \frac{2}{3\cdot 2^{k-1}-1}>0$ we have that all $\gamma_i>0$.

        \begin{proposition}\label{prop:eccparams}
            $1-\delta -\alpha_{i+1}+\gamma_i = \frac{\delta}{2}$.
        \end{proposition}
        \begin{proof}
            \begin{align*}
                \alpha_{i+1}-\gamma_i = \gamma_{i+1}-2\gamma_i = \left(\frac{1}{2^{k-i-2}}-\frac{2}{2^{k-i-1}}\right)(1-\frac{\delta}{2})+1-\frac{3\delta}{2} = 1-\frac{3\delta}{2}.
            \end{align*}
        \end{proof}

        \begin{claim}\label{clm:eccbase}
            If there exists a vertex $x\in A_{k-1}$ such that $d(\hat{w},x) \geq \delta(\epsilon(w)+M)$ then $d(w,x) \geq (1-\delta)R - \delta M$.
        \end{claim}
        The proof follows from the triangle inequality. Thus, if after running Dijkstra's from $A_{k-1}$ we don't have a good enough approximation, then $d(v_{k-2}, A_{k-1}) > \delta(\epsilon(w)+M) = \gamma_{k-1}(\epsilon(w)+M)$.

        \begin{claim} \label{clm:eccstep}
            Suppose $d(v_i, A_{i+1}) > \gamma_{i+1}(\epsilon(w) + M)$, then after running Dijkstra's from $B_i(v_i)$ we either have an estimate $\tilde{\epsilon}_w\geq (1-\delta)\epsilon(w)- \delta M$ or $d(v_{i-1}, A_i) > \gamma_i (\epsilon(w) + M)$.
        \end{claim}
        \begin{proof}
            If $d(w, v_i) \geq (1-\delta)\epsilon(w) - \delta M$ we are done. Otherwise, consider a shortest path between $v_i$ and $w$. Let $x$ be a vertex on this path such that $d(v_i, x) \leq \alpha_{i+1}(\epsilon(w)+M)$ and $d(x, w) \leq (1-\delta)\epsilon(w) - \delta M - \alpha_{i+1}(\epsilon(w)+M) + M = (1-\delta -\alpha_{i+1})(\epsilon(w) + M)$.

            If $d(x, A_i) > \gamma_i(\epsilon(w)+M)$ then we have that $d(v_{i-1}, A_i) > \gamma_i (\epsilon(w)+M)$. Otherwise, there exists a vertex $q\in A_i$ such that $d(x,q)\leq \gamma_i(\epsilon(w)+M)$ and so by the triangle inequality $d(v_i, q) \leq (\alpha_{i+1} + \gamma_i)(\epsilon(w)+M) = \gamma_{i+1}(\epsilon(w)+M)$. Therefore, $q\in B_i(v_i)$ and so we have computed Dijkstra's from it. Furthermore, by \autoref{prop:eccparams},
            \[
            d(w,q) \leq (1-\delta -\alpha_{i+1} + \gamma_i)(\epsilon(w)+M) = \frac{\delta}{2}(\epsilon(w)+M).
            \]
            Therefore, by the triangle inequality,
            \[
            \tilde{\epsilon}_w \geq \epsilon(q)-d(w,q) \geq \epsilon(w)-2d(w,q) \geq \epsilon(w) - \delta(\epsilon(w) + M).
            \]
            And so we will have found the desired approximation for every $w$.
        \end{proof}

        By applying \autoref{clm:eccstep} $(k-1)$ times and using \autoref{clm:eccbase} as the base case, we conclude that we either compute an estimate $\tilde{R}\leq (1+\delta)R + \delta M$ or have that $d(v_0, A_1) > \gamma_1 (R+M) = (1-\delta)(R+M)$.

        Now consider the shortest path between $c$ and $v_0$. There exists a vertex $x$ on this path such that $d(v_0, x) \leq (1-\delta)(R+M)$ and $d(x,c) \leq R - (1-\delta)(R+M) + M = \delta(R+M)$. Thus, $x\in B_0(v_0)$ and $\epsilon(x) \leq (1+\delta)R + \delta M$ and so we obtain the desired approximation.
    \end{proof}